\documentclass[conference]{IEEEtran}

\usepackage{cite}
\usepackage{amsmath,amsfonts,amssymb,amsthm}
\usepackage[usenames]{color}

\usepackage{bm}

\usepackage{xspace}

\usepackage{paralist}

\usepackage{comment}

\usepackage{graphicx}
\usepackage{epstopdf}
\DeclareGraphicsExtensions{.eps,.pdf}

\usepackage{subfig}
\usepackage[colorlinks,urlcolor=blue,citecolor=blue,linkcolor=blue]{hyperref}

\usepackage{tabularx}
\usepackage{array}

\usepackage{braket}

\usepackage{algorithm}
\usepackage{algpseudocode}

\usepackage[colorlinks,urlcolor=blue,citecolor=blue,linkcolor=blue]{hyperref}

\newtheorem{claim}{Claim}

\newcommand{\Sec}[1]{\hyperref[sec:#1]{\S\ref*{sec:#1}}} %
\newcommand{\Eqn}[1]{\hyperref[eq:#1]{(\ref*{eq:#1})}} %
\newcommand{\Fig}[1]{\hyperref[fig:#1]{Figure\,\ref*{fig:#1}}} %
\newcommand{\Tab}[1]{\hyperref[tab:#1]{Table\,\ref*{tab:#1}}} %
\newcommand{\Thm}[1]{\hyperref[thm:#1]{Thm.\,\ref*{thm:#1}}} %
\newcommand{\Lem}[1]{\hyperref[lem:#1]{Lem.\,\ref*{lem:#1}}} %
\newcommand{\Prop}[1]{\hyperref[prop:#1]{Prop.~\ref*{prop:#1}}} %
\newcommand{\Cor}[1]{\hyperref[cor:#1]{Cor.~\ref*{cor:#1}}} %
\newcommand{\Def}[1]{\hyperref[def:#1]{Defn.~\ref*{def:#1}}} %
\newcommand{\Alg}[1]{\hyperref[alg:#1]{Alg.~\ref*{alg:#1}}} %
\newcommand{\Ex}[1]{\hyperref[ex:#1]{Ex.~\ref*{ex:#1}}} %
\newcommand{\Clm}[1]{\hyperref[clm:#1]{Claim~\ref*{clm:#1}}} %

\newcommand{\din}[1][i]{d_{#1}^\leftarrow}
\newcommand{\dout}[1][i]{d_{#1}^\rightarrow}
\newcommand{\drec}[1][i]{d_{#1}^\leftrightarrow}
\newcommand{\dtin}[1][i]{d_{#1}^\Leftarrow}
\newcommand{\dtout}[1][i]{d_{#1}^\Rightarrow}
\newcommand{\dmax}{d_{\max}}

\newcommand{\ndin}[1][d]{n_{#1}^\leftarrow}
\newcommand{\ndout}[1][d]{n_{#1}^\rightarrow}
\newcommand{\ndrec}[1][d]{n_{#1}^\leftrightarrow}
\newcommand{\ndtin}[1][d]{n_{#1}^\Leftarrow}
\newcommand{\ndtout}[1][d]{n_{#1}^\Rightarrow}

\newcommand{\brec}{b^\leftrightarrow}
\newcommand{\bin}{b^\leftarrow}
\newcommand{\bout}{b^\rightarrow}
\newcommand{\btin}{b^\Leftarrow}
\newcommand{\btout}{b^\Rightarrow}

\newcommand{\Prob}[1]{\text{\rm Prob}\Set{ #1 }}

\newcommand{\TTF}[1]{\multicolumn{1}{|c|}{\bf #1}}
\newcommand{\TT}[1]{\multicolumn{1}{c|}{\bf #1}}
\newcommand{\TC}[1]{\multicolumn{1}{c|}{#1}}

\begin{document}

\title{A Scalable Null Model for Directed Graphs Matching All Degree Distributions: In, Out, and Reciprocal}
\author{\IEEEauthorblockN{Nurcan Durak,\    Tamara G.~Kolda,\   Ali Pinar, and     C. Seshadhri}
\IEEEauthorblockA{Sandia National Laboratories \\ Livermore, CA  USA\\
Email: nurcan.durak@gmail.com, \{tgkolda, apinar, scomand\}@sandia.gov}
}

\maketitle

\begin{abstract} 
Degree distributions are arguably the most important property of real world networks.
  The classic edge configuration model or Chung-Lu model can generate an undirected graph with any desired degree distribution.
  This serves as a good  \emph{null model} to compare algorithms or perform experimental studies.
Furthermore, there are scalable algorithms that implement
  these models and they are invaluable in the study of graphs. However, networks in the real-world
  are often directed, and have a significant proportion of \emph{reciprocal edges}.
  A stronger relation exists between two nodes when they each point to one another
  (\emph{reciprocal edge}) as compared to when only one points
  to the other (\emph{one-way edge}). Despite their importance, reciprocal
  edges have been disregarded by most directed graph
  models.  

  We propose a null model for directed graphs inspired by the Chung-Lu model that matches the in-, out-, and reciprocal-degree
  distributions of the real graphs. Our algorithm is scalable and requires $O(m)$ random numbers to generate
  a graph with $m$ edges. We perform a series of experiments on real datasets and compare
  with existing graph models.
\end{abstract}

\section{Introduction}

Ever since the seminal work of Barab\'{a}si and Albert~\cite{BaAl99}, Faloutsos et al.~\cite{FFF99}, Broder et al.~\cite{BrKu+00}, degree distributions are widely regarded
as a key feature of real-world networks. The heavy-tailed nature of these degree distributions
has been repeatedly observed in a wide variety of domains. 
One of the invaluable tools in analyzing heavy-tailed graphs is the ability to produce
a random or ``generic" graph with a desired degree distribution. 
The classic edge configuration~\cite{bc1978, b1980, MoRe98, n2003} does exactly that and 
is a common method for constructing such graphs. Chung and Lu~\cite{AiChLu00,ChLu02-2} give
more analyzable variants of this model.
MCMC methods based on random walks
are also used for this purpose~\cite{ktv,GkantsidisMMZ03}. 

These constructions are useful for testing algorithms and
comparing with existing models. It also helps in design of new algorithms. For example,
versions of the stochastic block model~\cite{BiCh09,KaNe11} used for community detection use Chung-Lu
type constructions for null models. The classic notion of modularity~\cite{GiNe02} measures deviations
from a Chung-Lu structure to measure community structure.
At a higher level, having a baseline model that accurately matches
the degree distribution informs us about other properties. Notably, work on the eigenvalue
distributions on Chung-Lu graphs~\cite{Mihail02,ChLuVu03} suggest the observations on so called ``eigenvalue power laws"
are simply a consequence of heavy tailed degree distributions.
For these reasons, we think of the edge-configuration or Chung-Lu constructions as \emph{null models}.

While all of this work has been extremely useful in advancing graph mining, it ignores the 
crucial property of \emph{direction} in networks. Most interaction, communication, web networks
are inherently directed, and the standard practice is to make these undirected.
Furthermore, directed networks exhibit \emph{reciprocity}, where some pairs of vertices
have edges in both directions connecting them. For example, in \Fig{sampleGraph},
there are two-way connections between some vertices. This indicates a much stronger connection between them.

Newman~\cite{newman02} introduces the reciprocity, $r$, which measures the density of reciprocal
edges in a network. It can be interpreted as the
probability of a random edge in a network being reciprocated. The
reciprocity is often high in social networks but is lower in information networks; see \Tab{networks}.
It was observed that high reciprocity leads to faster spread of viruses or news~\cite{newman02, GaLo04}. The importance of reciprocal edges
is underscored by a study of formation order of these edges~\cite{MiKoGuDrBh08}.
In the Flickr network (which has 68\% reciprocal edges), 83\% of
all reciprocal edges are created within 48 hours after the initial
edge creation. The Twitter network has 22.1\% of the reciprocal edges~\cite{KwLePaMo:10}.  Reciprocity also plays an important role in interactions in massive multiplayer online games~\cite{SuSi13}.
All these studies show that reciprocal edges are quite special, and provide important information about the
social processes underlying these graphs. But all graph models and constructions completely
ignore these edges.

A key concern with graph generation is simple construction and scalability, as we may want
test instances with millions (and more) edges. A key feature for  a null model is its scalability and its  ability to quickly produce 
a large graph that matches degree distributions.

\begin{figure}[h]
  \centering
    \includegraphics[scale = 0.5]{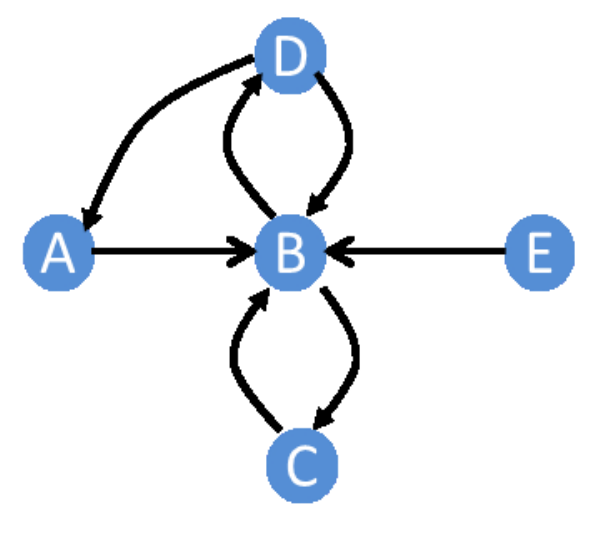}
  \caption{A directed graph with reciprocal (e.g., B-D) and one-way (e.g., D-A) edges.}
  \label{fig:sampleGraph}
\end{figure}

\subsection{Contributions}

For a directed graph, there are three distinct degree distributions associated with it:
the in-degree, the out-degree, and the reciprocal degree distribution.
The last can be thought of as the degree distribution of the undirected subgraph
obtained by only taking reciprocal edges. A good null model, along the lines
of the configuration model or Chung-Lu, must match all three of these.
We design the \emph{Fast Reciprocal Directed} (FRD) graph generator that
does exactly that.

\begin{compactitem}
\item The \emph{Fast Reciprocal Directed} (FRD) graph generator takes as input
in-,out-, and reciprocal degree distributions, and produces a random graph matching
these. It can be thought of as a generalization of the Chung-Lu model for this setting.
We provide a series of empirical results showing how it matches these degree distributions
for real datasets.
\item Our algorithm is \emph{fast and scalable}. It only requires some minimal preprocessing
and the generation of $O(m)$ random numbers. It takes less than a minute to generate a graph with multi-million nodes and
  edges, faster than any comparable models.
\item We compare FRD's degree distributions fits with existing directed graph models.
At some level, this is not a fair comparison, since we do not consider our generator to be realistic
(while competing methods attempt to match other important graph properties and mimic real world processes).
Our model is meant to be a baseline or null model that matches degree distributions. 
But our comparisons with realistic graph models are quite illuminating. Despite
the large number of reciprocal edges in real networks, none of the other models come even marginally close
to matching the reciprocal degree distribution.
\item As an aside, we explain why the number of degree-1 nodes is much lower than intended in
  Chung-Lu like models \cite{ChLu02, SeKoPi12} and propose a solution to obtain a better match
  for the degree-1 vertices. This fix is incorporated in the FRD generator.
\end{compactitem}

\section{Related Work}

As mentioned earlier, edge configuration models have a long history. Miller and Hagberg~\cite{MiHa11} discuss faster algorithms for implementing
Chung-Lu, while Seshadhri et al.~\cite{KoPiPlSe13} discuss a different parallel version.
A directed version of the edge configuration model together with mathematical analyses
of connected component structure was given in~\cite{MeNePo06}.
Our work is related to this construction.

Reciprocal edges are not taken into account by most common graph models.
The Forest Fire (FF) model~\cite{LeKlFa05} and Stochastic Kronecker Graph (SKG)
model \cite{LeFa07,LeChKlFa10} are often used to generate graphs, and do
produce directed graphs. They can match in- and out-degree distributions reasonably well,
and we use these models for comparisons.

Most common graph models (e.g., preferential attachment~\cite{BaAl99}, edge copying model\cite{KlKu99}, forest fire~\cite{LeKlFa05}) 
produce directed graphs incrementally to imitate the growth of graphs.
They produce heavy-tailed in- and out- degree distributions, but almost no reciprocal edges. Furthermore,
they are not scalable to millions of nodes and billions of edges.
The Stochastic Kronecker Graph
model~\cite{LeFa07,LeChKlFa10} is scalable, but is also unable to produce reciprocity.
In this study, we compare our results with the Forest Fire (FF) model and Stochastic Kronecker Graph (SKG) model.

A notable exception is work of Zlatic et al.~\cite{Zl09, ZlSt11} that generalizes Preferential
Attachment (PA) using reciprocal edges. Unfortunately, it is not scalable and does
not match out-degree distributions (in their experiments). Another variant of PA~\cite{BoBoChRi03}
does allow edges between existing nodes (thereby introducing some reciprocity), but the model
is not meant to really match real data.

\section{The Fast Reciprocal Directed Null Model}
\label{sec:model}

We first introduce some notation. Given a directed
graph $G$, let $n$ be the number of nodes and $m$ be the number of
directed edges. For instance, in \Fig{sampleGraph}, $n=5$ and $m=7$.
We divide the edges into three types:
\begin{compactitem}
\item $\drec$ = reciprocal degree (each reciprocal edge corresponds to a \emph{pair} of directed edges),
\item $\din$ = in-degree (excluding reciprocal edges), and
\item $\dout$ = out-degree (excluding reciprocal edges).
\end{compactitem}
We also define the \emph{total} in- and out- degrees, which include
the reciprocal edges, i.e.,
\begin{compactitem}
\item $\dtin = \din + \drec$ = total in-degree, and
\item $\dtout = \dout + \drec$ = total out-degree.
\end{compactitem}
Most directed graph models consider only the total in- and out-degrees, ignoring reciprocity.
As an example of these measures, node B in \Fig{sampleGraph}
has $\drec[B]=2$, $\din[B] = 2$, $\dout[B] = 0$, $\dtin[B]=4$, and $\dtout[B]=2$.

We may also assemble corresponding degree distributions, as
follows. For any $d=0,1,\dots$, define
\begin{compactitem}
\item $\ndrec$ = Number of nodes with reciprocal-degree $d$,
\item $\ndin$ = Number of nodes with in-degree $d$,
\item $\ndout$ = Number of nodes with out-degree $d$,
\item $\ndtin$ = Number of nodes with total-in-degree $d$, and
\item $\ndtout$ = Number of nodes with total-out-degree $d$.
\end{compactitem}
Let $d_{\max}$ be the maximum of all possible degrees.
Then we can express $n$ and $m$ as
\begin{align*}
  n &= \sum_{d=0}^{\dmax} {\ndin} = \sum_{d=0}^{\dmax} {\ndout}  = \sum_{d=0}^{\dmax} {\ndrec}, \\
  m &= \sum_{d=1}^{\dmax} d \cdot {\ndin} + d \cdot {\ndrec}
  = \sum_{d=1}^{\dmax} d \cdot {\ndout} + d \cdot {\ndrec}.
\end{align*}
The reciprocity of a graph \cite{newman02} is
\begin{displaymath}
   r = \frac{\text{\# reciprocated edges}}{\text{\# edges}}
   =\frac{\sum_{d=1}^{\dmax} d \cdot {\ndrec} }{m}.
\end{displaymath}

We will present an extension of the Chung-Lu model that accounts for in- and out-degrees.
This will be a part of the final FRD generator.

\subsection{The Fast Directed Generator}

In this first step, we consider only the total in- and out-degrees and
ignore reciprocity. This can be thought of as a fast implementation
of the directed edge configuration model in~\cite{MeNePo06}.
We extend the Fast Chung-Lu (FCL) algorithm for undirected graphs \cite{SeKoPi12}.  This
is based on the idea that each edge creation can be done
independently if the degree distribution is given. The FCL reduces the
complexity of the CL model from $O(n^2)$ to $O(m)$, and the
same can be done in the directed case.

In the Chung-Lu model \cite{ChLu02}, after $m$ insertions (and assuming $\dtout \dtin[j] < m$ for all $i,j$) the probability of edge~$(i,j)$ is 
\begin{displaymath}
  p_{ij} = \frac{\dtout \dtin[j]}{m}.
\end{displaymath}

The naive approach flips a coin for each edge independently. The
``fast'' approach flips a coin to pick each endpoint.  The probability
of picking node $i$ as the source is proportional to $\dtout$ and the
probability of picking node $j$ as the destination is proportional to
$\dtin[j]$.

Our implementation works as described in \Alg{fdmodel}.  We first pick
all the source nodes and then all the sink nodes using the weighted
vertex selection described in \Alg{vertexselect}.  If we want 500
nodes with out-degree of $2$, for example, we create a ``degree-2 pool''
of 500 vertices and pick from it a total of 1000 times in expectation by
doing weighted sampling of the pools.  Within the pool, we pick a
vertex uniformly at random with the further expectation that each
vertex in the pool will be picked 2 times on average. In
\Alg{vertexselect}, the pool of degree-$d$ vertices is denoted by
$\mathcal{P}_d$ and the likelihood that the $d$th pool is selected is
denoted by $w_d$. In all cases except $d=1$, the size of the pool is
defined by the number of vertices of that degree and the weight of the
pool is the number of edges that should be in that pool. The one
exception is the degree-1 pool, which has a \emph{blowup} factor
$b$. For now, assume $b=1$; we explain its importance further on in
\Sec{blowup}.  At the end of \Alg{vertexselect}, we randomly relabel
the vertices so there is no correlation between the degree and vertex
identifier.

The FD method can produce repeat edges, unlike the naive version that
flips $n^2$ weighted coins (one per edge). Nevertheless, this has not
been a major problem in our experience. Another alternative to
\Alg{vertexselect} is to put $d$ copies of each degree-$d$ vertex into
a long array and then randomly permute it---this is the approach of
the \emph{edge configuration} model. This gives the \emph{exact}
specified degree distribution (excepting possible repeats) by using a random permutation of a length $m_*$ array. This would
produce very similar results to what we show here, and is certainly a
viable alternative. We also mention an alternate way of generating Chung-Lu graphs that could be adapted for the directed case \cite{MiHa11}.

\begin{algorithm}[ht]
  \caption{Fast Directed Graph Model}
  \label{alg:fdmodel}
  \begin{algorithmic}
    \Procedure{FDModel}{$G$,$\btin$,$\btout$}
    \State Calculate $\set{\ndtin}$ and $\set{\ndtout}$ for $G$
    \State $\set{i_k} \gets \textsc{VertexSelect}(\set{\ndtout}, \btout)$
    \State $\set{j_k} \gets \textsc{VertexSelect}(\set{\ndtin}, \btin)$
    \State $E \gets \set{(i_k, j_k)}$
    \State \emph{Remove self-links and duplicates from $E$}
    \State \Return{$E$}
    \EndProcedure
  \end{algorithmic}
\end{algorithm}

\begin{algorithm}[ht]
  \caption{Weighted Vertex Selection}
  \label{alg:vertexselect}
  \begin{algorithmic}
    \Procedure{VertexSelect}{$\set{n_d}$, $b$}
    \State $n \gets \sum_{d=0}^{\dmax} n_d$
    \State $n_* \gets b \cdot n_1 + \sum_{d=2}^{\dmax} n_d$
    \State $m \gets \sum_{d=1}^{\dmax} d \cdot n_d$
    \State $\mathcal{P} = \set{1,\dots,n_*}$
    \ForAll{$d=1,\dots,\dmax$}
    \State $w_d \gets d \cdot n_d / m$
    \If{$d>1$}
    \State $\mathcal{P}_d \gets $\text{$n_d$ vertices from $\mathcal{P}$}
    \Else
    \State $\mathcal{P}_1 \gets $\text{$b \cdot n_1$ vertices from $\mathcal{P}$}
    \EndIf
    \State $\mathcal{P} \gets \mathcal{P} \setminus \mathcal{P}_d$
    \EndFor
    \ForAll{$k=1,\dots,m$}
    \State $\hat d_k \gets$ Random degree in $\set{1,\dots,\dmax}$,
    \State \phantom{$\hat d_k \gets$ }proportional to weights $\set{w_d}$
    \State $i_k \gets$ Uniform random vertex in $\mathcal{P}_{\hat d_k}$
    \EndFor
    \State $\mathcal{P} \gets$ unique indices in $\set{i_k}_{k=1}^m$
    \State $\pi \gets $ Random mapping from $\mathcal{P}$ to $\set{1,\dots,n}$
    \State \Return{$\set{\pi(i_k)}_{k=1}^{m}$}
    \EndProcedure
  \end{algorithmic}
\end{algorithm}

\subsection{Introducing reciprocity}

The FD model generates a directed graph and matches to the total in-
and out-degree distributions. However, it produces virtually no
reciprocal edges. The FRD null model explicitly introduces reciprocity
using an undirected model and uses FD for remaining directed edges.
We blend the two schemes in one model. In this case, we explicitly consider the three distributions,
$\set{\ndrec}$, $\set{\ndin}$, and $\set{\ndout}$.
The method is presented in \Alg{frdmodel}.

\begin{algorithm}[ht]
  \caption{Fast Reciprocal Directed Graph Model}
  \label{alg:frdmodel}
  \begin{algorithmic}
    \Procedure{FDModel}{$G$,$\brec$, $\bin$,$\bout$}
    \State Calculate $\set{\ndrec}$, $\set{\ndin}$, and $\set{\ndout}$ for $G$
    \State $\set{i_k} \gets \textsc{VertexSelect}(\set{\frac{1}{2} \ndrec}, \brec)$
    \State $\set{j_k} \gets \textsc{VertexSelect}(\set{\frac{1}{2} \ndrec}, \brec)$
    \State $E_1 \gets \set{(i_k, j_k), (j_k,i_k)}$
    \State $\set{i_l} \gets \textsc{VertexSelect}(\set{\ndout}, \bout)$
    \State $\set{j_l} \gets \textsc{VertexSelect}(\set{\ndin}, \bin)$
    \State $E_2 \gets \set{(i_l, j_l)}$
    \State $E \gets E_1 \cup E_2$
    \State \emph{Remove self-links and duplicates from $E$}
    \State \Return{$E$}
    \EndProcedure
  \end{algorithmic}
\end{algorithm}

\subsection{Fixing the Number of Degree-1 Nodes}
\label{sec:blowup}
Below, we present our arguments for the case of the in-degree, but the same arguments
applied to out or reciprocal degrees. 
We use just the notation $d$ to denote the in-degree, for simplicity.

If we run \textsc{VertexSelect} (\Alg{vertexselect}) repeatedly,
always assigning the same ids to each vertex pool and omitting the
random relabeling ($\pi$) at the end, each node will get its desired
in-degree \emph{on average across multiple runs}.  For any single run,
however, this will not be the case. In fact, the degrees are Poisson
distributed.

\begin{claim}
  The probability that a vertex $v$ in pool $\mathcal{P}_d$ is selected $x$ times is
  \begin{displaymath}
    \Prob{ \text{$v$ selected $x$ times} | v \in \mathcal{P}_d}
    = \frac{d^x e^{-d}}{x!}.
  \end{displaymath}
\end{claim}

This claim is easy to see. We expect that pool $\mathcal{P}_d$ will be
selected $w_d = d\cdot n_d$ times. Therefore, each element of
$\mathcal{P}_d$ will be selected an average of $d$ times, so that is
the Poisson parameter. (There may be some small variance in the number
of times that each pool is selected, but the variance should be small
enough not to greatly impact the average degree.)

The effect of the Poisson distribution is particularly noticeable in
the pool of degree-1 nodes where the probability that a node in
$\mathcal{P}_1$ has in-degree $x=1$ is only 36\%. An additional
36\% will have an in-degree of $x=0$ and the remaining 28\% will an
in-degree of $x \geq 2$.  Of course, there will be some contributions
from the other pools, e.g., $\mathcal{P}_2$ will produce 27\% degree-1
nodes. However, in a power law degree distribution, $n_2 \ll n_1$ so
its contribution is small. Nevertheless, we can calculate the expected
number of degree-$x$ nodes by summing over the contributions across
all degrees pools.
\begin{claim}
  Let $n_x'$ denote the number of nodes that are selected exactly $x$
  times. Then
  \begin{displaymath}
    \mathbb{E}(n_x') = \sum_d n_d \frac{d^x e^{-d}}{x!}.
  \end{displaymath}
\end{claim}

Again, the claim is easy to see and so the proof is omitted.

For many real-world distributions, $n_1' \ll n_1$.  We propose a
workaround to this problem --- we would like to reduce the number of
nodes in $\mathcal{P}_1$ that are selected multiple times. To do this,
we increase the size of the pool via a \emph{blowup factor} $b$, which
is used as follows.  Let $\mathcal{P}_1$ contain $b \cdot n_1$
nodes. The weight of the pool will not change, meaning that it will
still be selected $n_1$ times. Therefore, we may make the following
claim.

\begin{claim}
  The probability that a vertex $v$ in pool $\mathcal{P}_1$ with $b
  \cdot n_1$ elements is selected $x$ times is
  \begin{displaymath}
    \Prob{ \text{$v$ selected $x$ times} | v \in \mathcal{P}_1}
    = e^{-1/b}/{(b^x \cdot x!)}.
  \end{displaymath}
  Furthermore, the expected number of nodes in $\mathcal{P}_1$ that are
  selected exactly one time is $n_1 \cdot e^{-1/b}$. Hence, letting
  $n_x'$ denote the number of nodes that are selected exactly $x$
  times, we have
  \begin{displaymath}
    \mathbb{E}(n_x') = n_1 \cdot \frac{e^{-1/b}}{b^{x-1} \cdot x!} + \sum_{d>1} n_d \frac{d^x e^{-d}}{x!}.
  \end{displaymath}
\end{claim}
\begin{proof}
  We still pick pool $\mathcal{P}_1$ a total of $n_1$ times, so that
  average (i.e., the Poisson parameter) for this pool is now reduced
  to $n_1/(n_1 \cdot b) = 1/b$ since there are $b \cdot n_1$ elements.

  The next equation comes from the fact that there are $b \cdot n_1$
  nodes in the pool, so we multiply the number of nodes with the probability of
  being picked $x$ times with $x=1$ to determine the expected number. 

  Finally, the revised expectation comes from changing the formula for
  the first pool to account for the enlarged pool size.
\end{proof}

If we choose, for example, $b=10$, then we can expect that $0.9 \cdot
n_1$ nodes in $\mathcal{P}_1$ to be selected exactly one time.  We
show an example of the impact of this modification in
\Fig{degree1fix}, where we show the total in-degree for soc-Epinions1
with and without a blowup factor of $b=10$.  The degrees are
logarithmicly binned and summed.  Note that the match for the number
of degree-1 nodes is improved, but there is a small penalty in the match
for degree-2 nodes. We use $b=10$ in all experiments reported in this paper.

\begin{figure}[htb]
  \centering
  \includegraphics[width=2in]{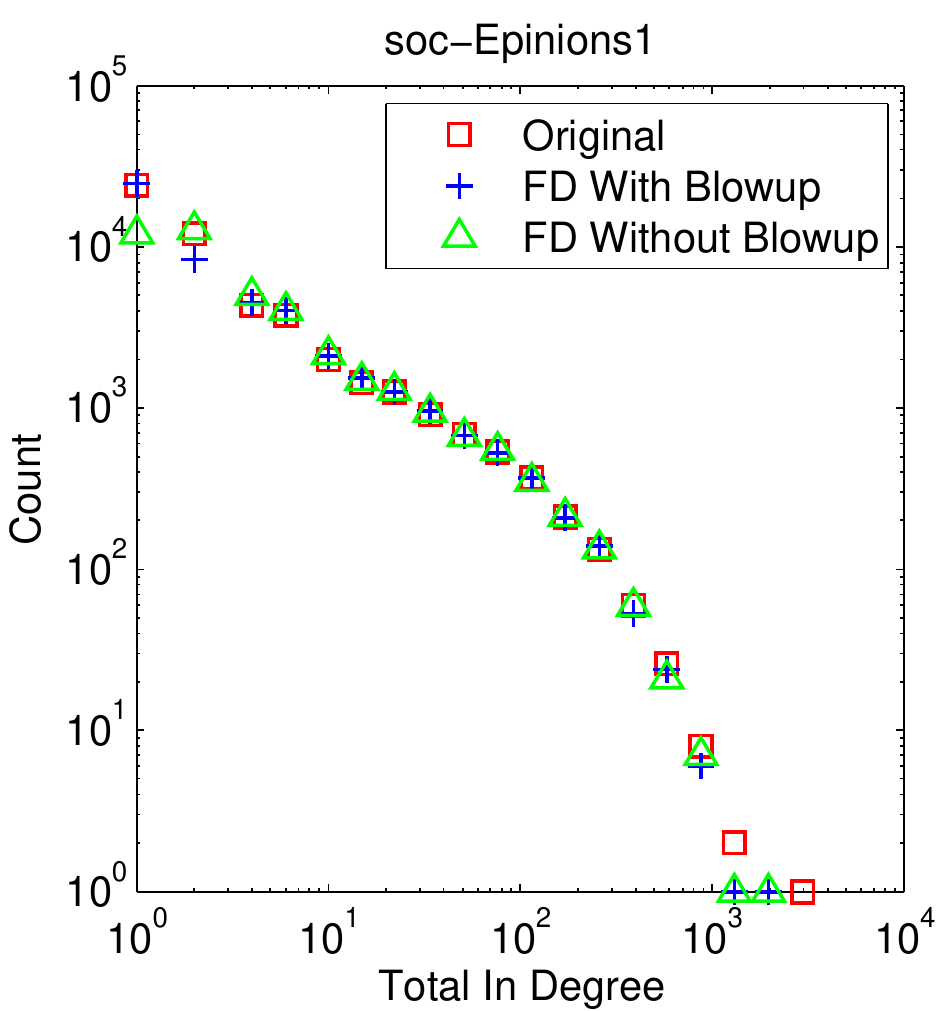}
  \caption{Example of in-degree distribution with and without blowup factor. Note that the model with the blow-up factor matches degree-1 nodes precisely, however, the model without blow-up generates only half of the degree-1 nodes in the original graph.}
  \label{fig:degree1fix}
\end{figure}

\section{Experimental Studies}

We test our models on various directed networks such as
citation (cit-HepPh), web (web-NotreDame), and social (soc-Epinions1,
soc-LiveJournal) \cite{Snap}. We also test our models on large scale
graphs coming from online social networks (youtube, flickr,
liveJournal) \cite{MiMaGu07}. We list the
attributes of the networks in \Tab{networks} after removing self-links
and making the graph unweighted (simple). As expected, the reciprocity $r$ is very low in the citation network.
We elaborate how we fit the models to the real networks below.

\begin{table*}[htb]
  \caption{Networks used in this study. $r$ is the reciprocity, $p_f$ is the forward burning parameter for FF, and the last column is the SKG initiator matrix.}
\begin{center}\small
\begin{tabular}{|l | r | r |r|r| r| r | }
\hline
\TTF{Graph Name}	& \TT{Nodes}  & \TT{Edges} & \TT{Rec.~Edges} & \TT{$r$}  & \TT{$p_f$} & \TT{SKG initiator} \\ \hline
cit-HepPh \cite{Snap} 	& 34K & 421K	&1,314	&0.003 & 0.37& [0.990,0.440;0.347,0.538] \cite{LeChKlFa10}\\
\hline
soc-Epinions1 \cite{Snap} &76K &508K &206K     &0.405 & 0.346 &  [0.999,0.532;0.480,0.129] \cite{LeChKlFa10}\\
\hline
web-NotreDame \cite{Snap} 	& 325K	&1,469K &  759K & 0.517 &0.355 & [0.999,0.414;0.453,0.229] \cite{LeChKlFa10} \\
\hline
soc-LiveJournal \cite{Snap} &4,847K&	68,475K  & 32,434K & 0.632 & 0.358 & [0.896,0.597;0.597,0.099] \cite{Xu10}\\
\hline
youtube \cite{MiMaGu07}    &1,157K	&4,945K	&3,909K	&0.791 & 0.335& ---\\
\hline
flickr \cite{MiMaGu07}	    &1,861K	&22,613K	&14,117K	&0.624 & 0.355& ---\\
\hline
LiveJournal \cite{MiMaGu07}	&5,284K	&77,402K	&56,920K	&0.735 & 0.355 & --- \\
\hline
 \end{tabular}
\end{center}
\label{tab:networks}
\end{table*}

\paragraph{ Fast Directed (FD) and Fast Reciprocal Directed (FRD)}
  This only requires the appropriate degree distributions
  of the input graphs. We used a blowup factor of $b=10$ in all cases.

\paragraph{Forest Fire (FF)} We provide the number of nodes $n$, and
  the forward and backward burning probabilities $p_f$ and $p_b$ to
  the SNAP software \cite{Snap}. To fit FF, we picked  parameters  
  that best match the number of edges in the real networks. For each target graph,
  we search a range of values by incrementing $p_f$ value by $\delta
  p= 0.001$ in range [0.2-0.5] to find the best parameters, which 
  are reported in \Tab{networks}. We set $p_b=0.32$ as described
  in \cite{LeKlFa05}.

\paragraph{Stochastic Kronecker Graphs (SKG)} We
  use the initiator matrices reported by previous studies: \cite{LeChKlFa10} for
  cit-HepPh, soc-Epinions, and web-NotreDame and \cite{Xu10} for
  soc-LiveJournal. We attempted to generate initiator matrices for large graphs using \cite{Snap}, but
the program did not terminate within twenty-four hours. Therefore, we only fit SKG to the networks obtained from SNAP\cite{Snap} data warehouse. We set the size
  of the final adjacency matrix as $2^{\lceil \log_2(n) \rceil}$, where $n$ is the number of nodes in the real graph.

We generate all the models in a Linux machine with 12GB memory and
Intel Xeon 2.7 Ghz processor. The FD and FRD methods were
implemented in MATLAB. For  SKG and Forest Fire, we used  the C++
implementations in \cite{Snap}.  Graph generation time for each
model is listed in \Tab{time}. For fair comparison, we do not include I/O times. Among all of the results, FD and FRD
are the fastest, in that order. SKG is little bit slower than both FD and
FRD models. The forest fire is the slowest even though C++  codes are typically  much
faster than MATLAB codes.

\begin{table}[htb]
\caption{Graph generation times}
\begin{center}\small
\begin{tabular}{|l | r | r |r| r| }
\hline
\TTF{Graph Name}	& \TT{SKG}  & \TT{FD} 	&\TT{FRD}  &	\TT{FF} \\  \hline
cit-HepPh &2.17s &0.16s &0.19s & 18.80s\\
\hline
soc-Epinions  &1.53s  &0.29s & 0.41s & 6.73s\\
\hline
web-NotreDame  &4.95s & 0.56s & 0.62s & 29.66s \\
\hline
soc-LiveJournal & 6m51s  & 31.15s & 41.75s &2h28m32s \\
\hline
youtube  &\TC{---} &2.16s & 2.53s & 2m22s\\
\hline
flickr  &\TC{---} &10.30s & 12.20s & 1h11m2s\\
\hline
liveJournal  & \TC{---} &  35.30s & 59.98s & 8h30m18s\\
\hline
 \end{tabular}
\end{center}
\label{tab:time}
\end{table}

We analyze the number of reciprocal edges generated by each model in \Tab{reciprocal}. The FF model cannot generate any reciprocal edges. The FD model can generate a few random reciprocal edges but their number is negligible. The SKG model generates some reciprocal edges; yet a negligible amount. The FRD model performs the best and generate expected amount of reciprocal edges.

\begin{table}[htb]
\caption{Number of reciprocal edges created by each model}
\begin{center}\small
\begin{tabular}{|l | r | r |r| r| r|  }
\hline
\TTF{Graph Name}	& \TT{Orig.} & \TT{SKG}  & \TT{FD} 	&\TT{FRD}  &\TT{FF} \\ \hline 
cit-HepPh  &1,314 & 1,302 &160  &1,454 & 0 \\ \hline
 soc-Epinions1 &206K & 1,264 &114 & 205K& 0\\  \hline
 web-NotreDame  &759K& 766&28 & 757K& 0 \\  \hline
 soc-LiveJournal &32,434K& 16,520 &172& 32,432K& 0\\  \hline
 youtube &3,909K& \TC{---} & 18 & 3,873K& 0\\  \hline
 flickr &14,117K& \TC{---} & 222&14,032K & 0\\  \hline
 liveJournal &56,920K & \TC{---}&262 &56,454K & 0\\ \hline
\end{tabular}
\end{center}
\label{tab:reciprocal}
\end{table}

\setcounter{topnumber}{3}
\setcounter{dbltopnumber}{3}
\renewcommand{\topfraction}{.9}
\renewcommand{\dbltopfraction}{.9}

\begin{figure*}[htb]
  \centering
  \subfloat{\label{fig:irdd-soc-Epinions1}
  \includegraphics[width=1.9in,trim=0 0 0 0]{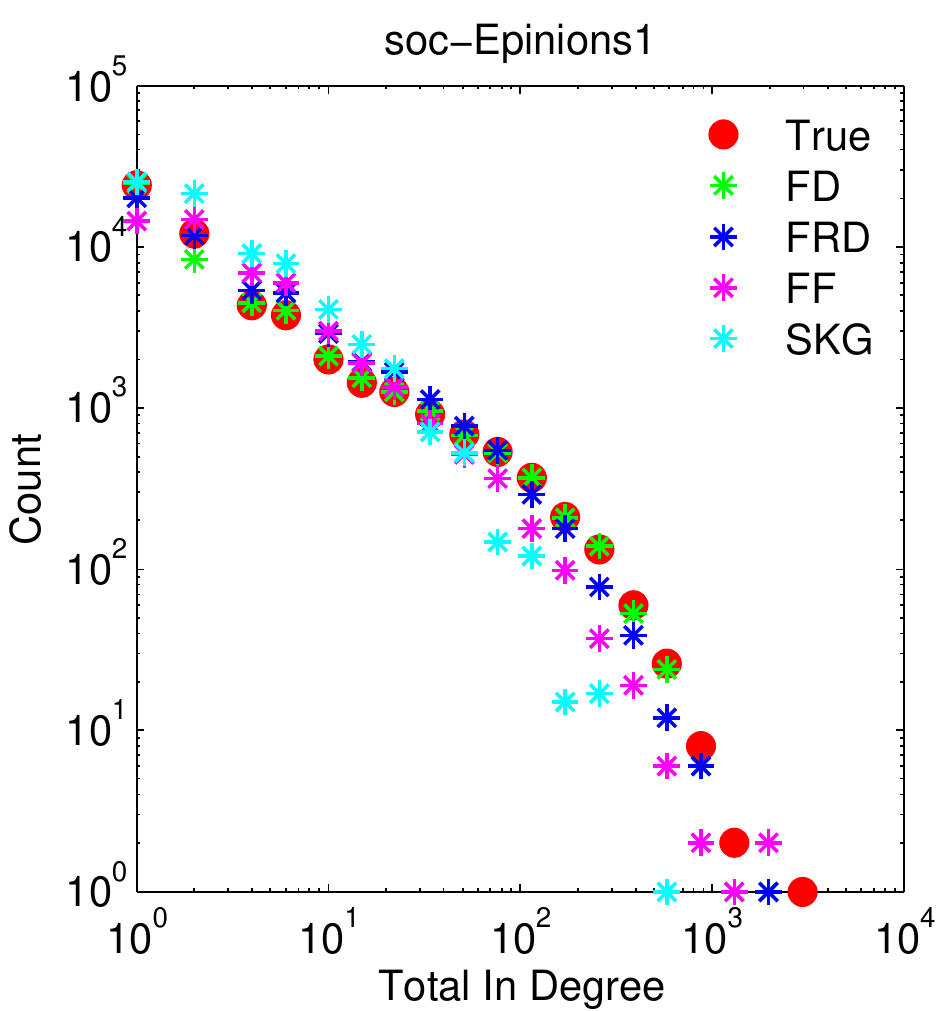}
  }
  \subfloat{\label{fig:ordd-soc-Epinions1}
    \includegraphics[width=1.9in,trim=0 0 0 0]{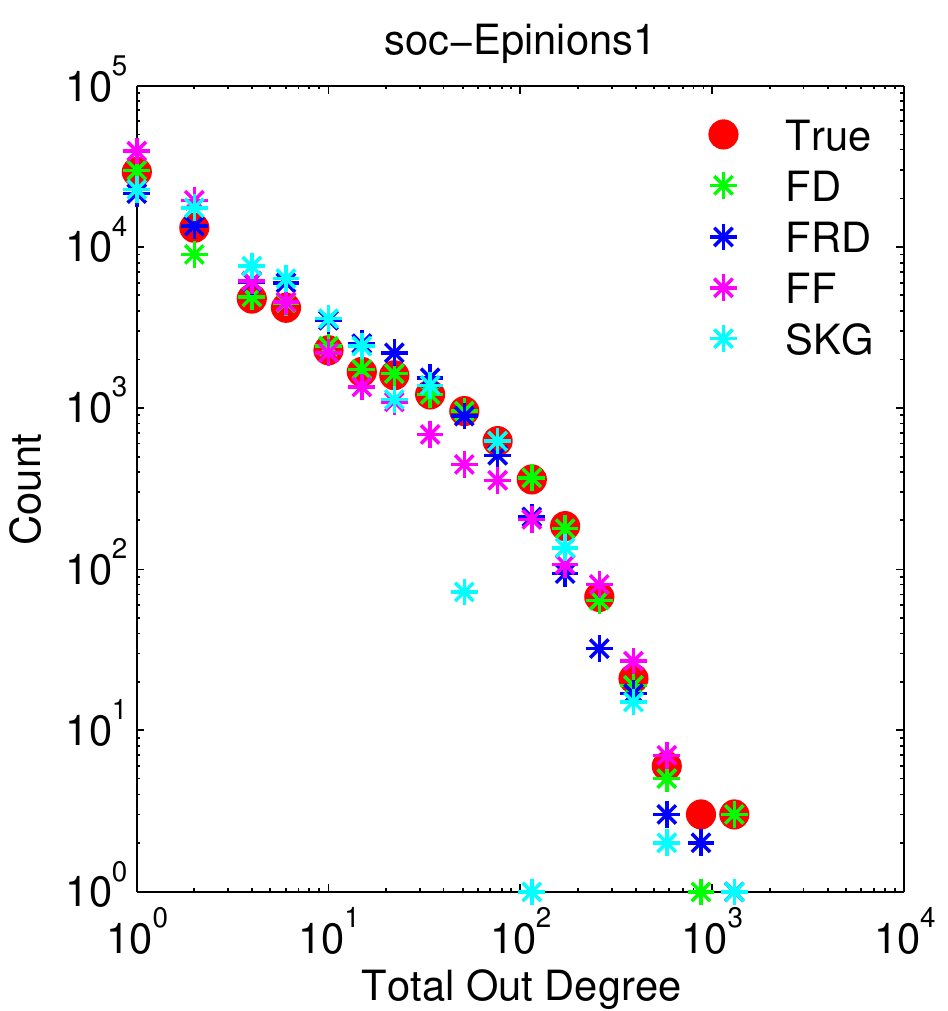}
  }
  \subfloat{\label{fig:rdd-soc-Epinions1}
    \includegraphics[width=1.9in,trim=0 0 0 0]{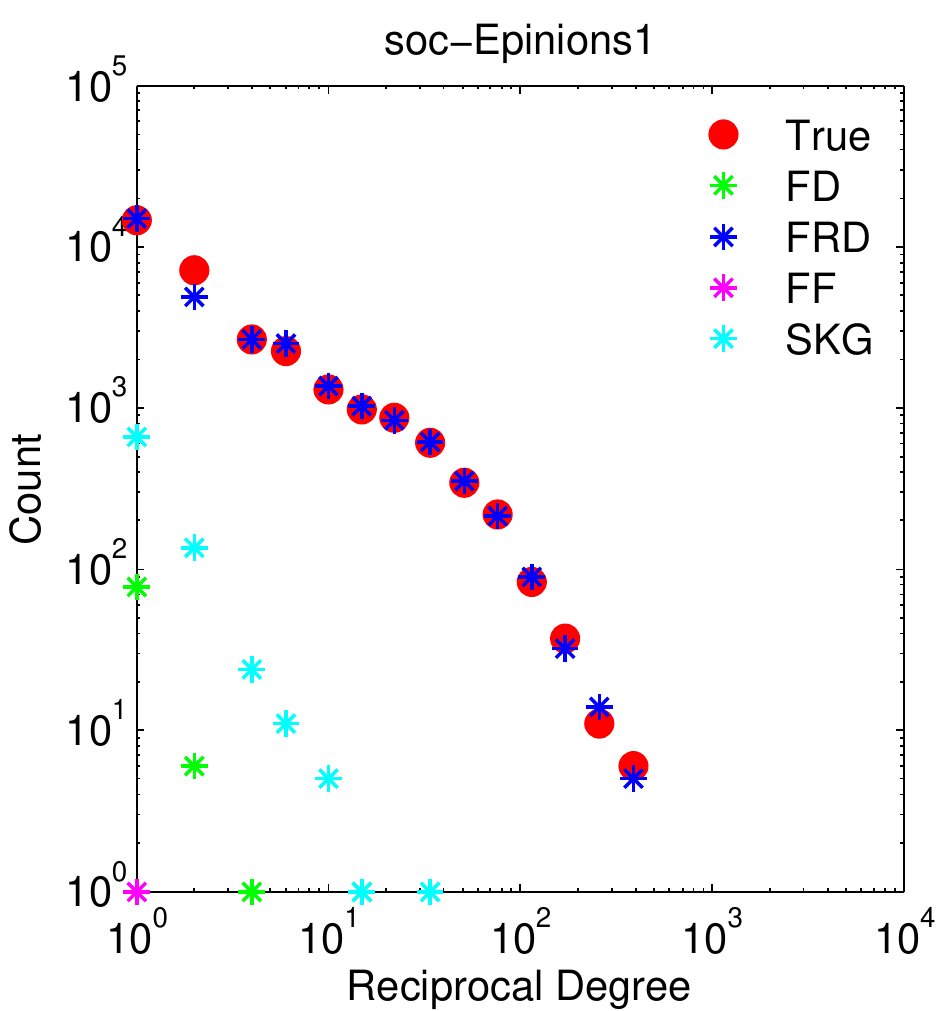}
  }
  \caption{Comparisons of degree distributions produced by various models for graph soc-Epinions1.}
  \label{fig:soc-Epinions1}
\end{figure*}
\begin{figure*}[htb]
  \centering
  \subfloat{\label{fig:irdd-soc-LiveJournal}
  \includegraphics[width=1.9in,trim=0 0 0 0]{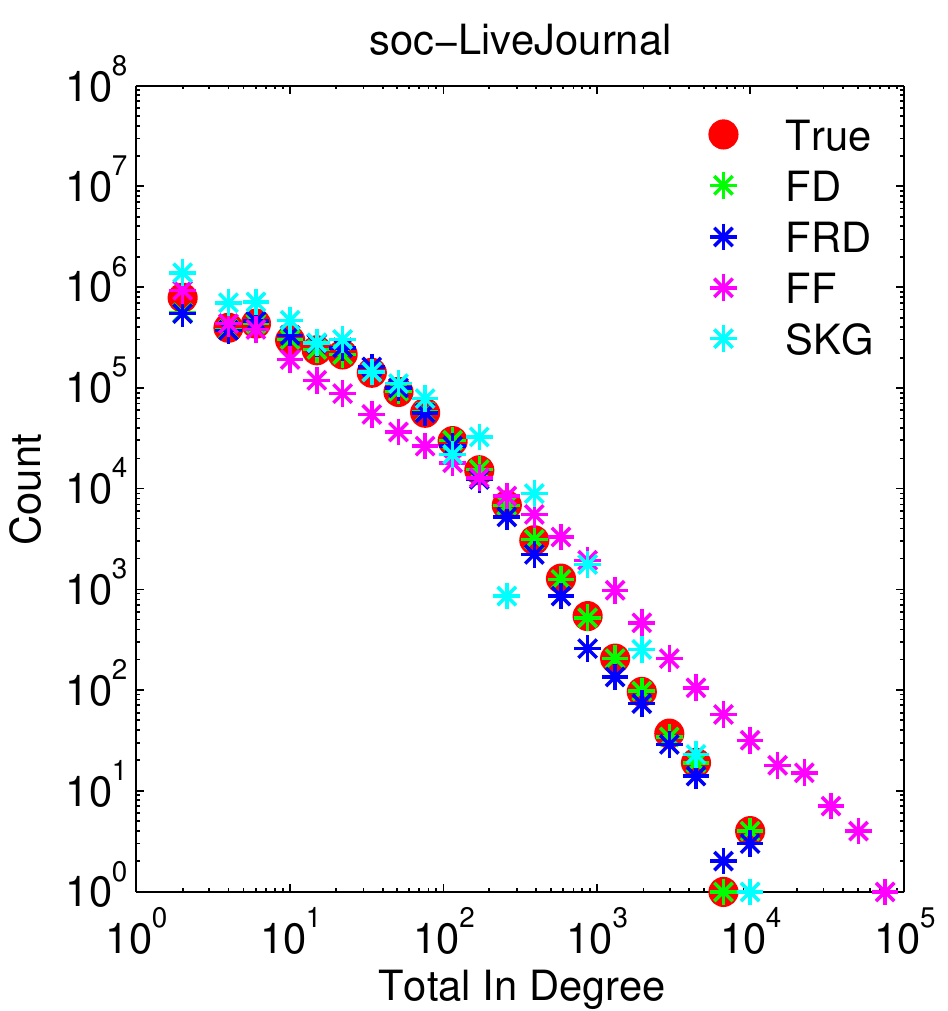}
  }
  \subfloat{\label{fig:ordd-soc-LiveJournal}
    \includegraphics[width=1.9in,trim=0 0 0 0]{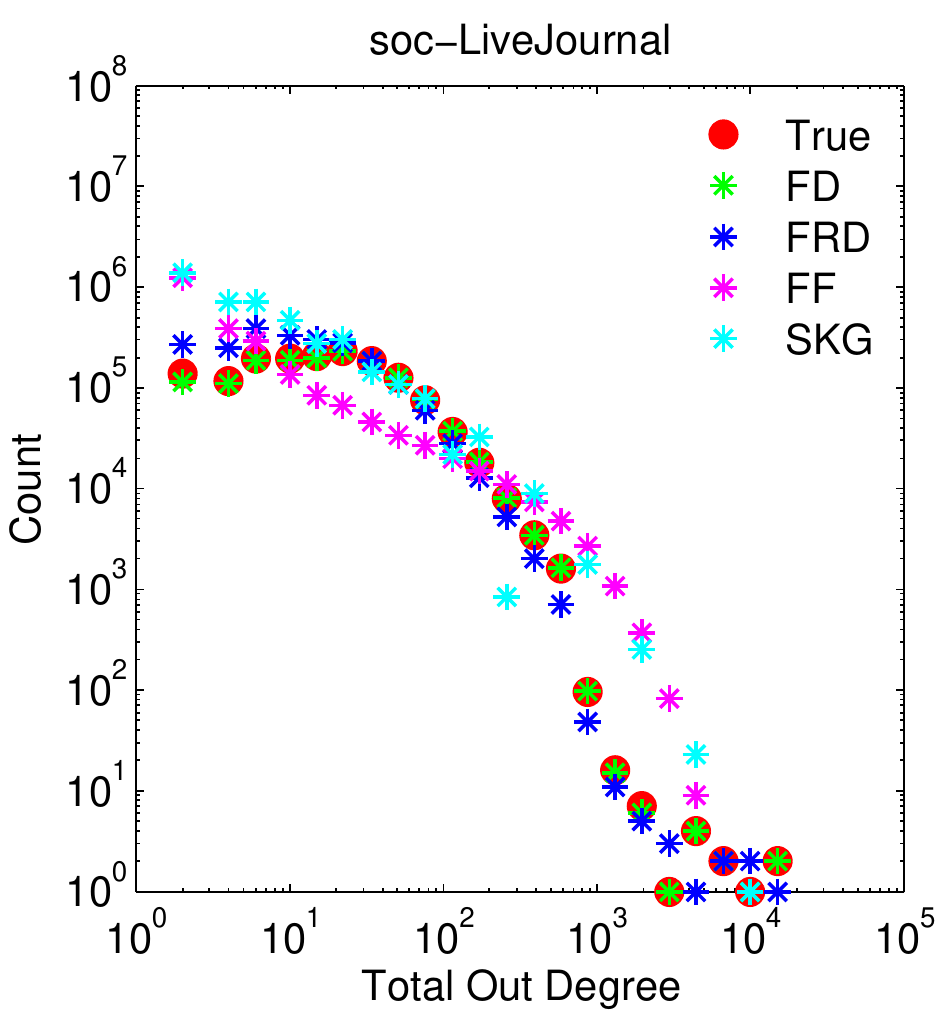}
  }
  \subfloat{\label{fig:rdd-soc-LiveJournal}
    \includegraphics[width=1.9in,trim=0 0 0 0]{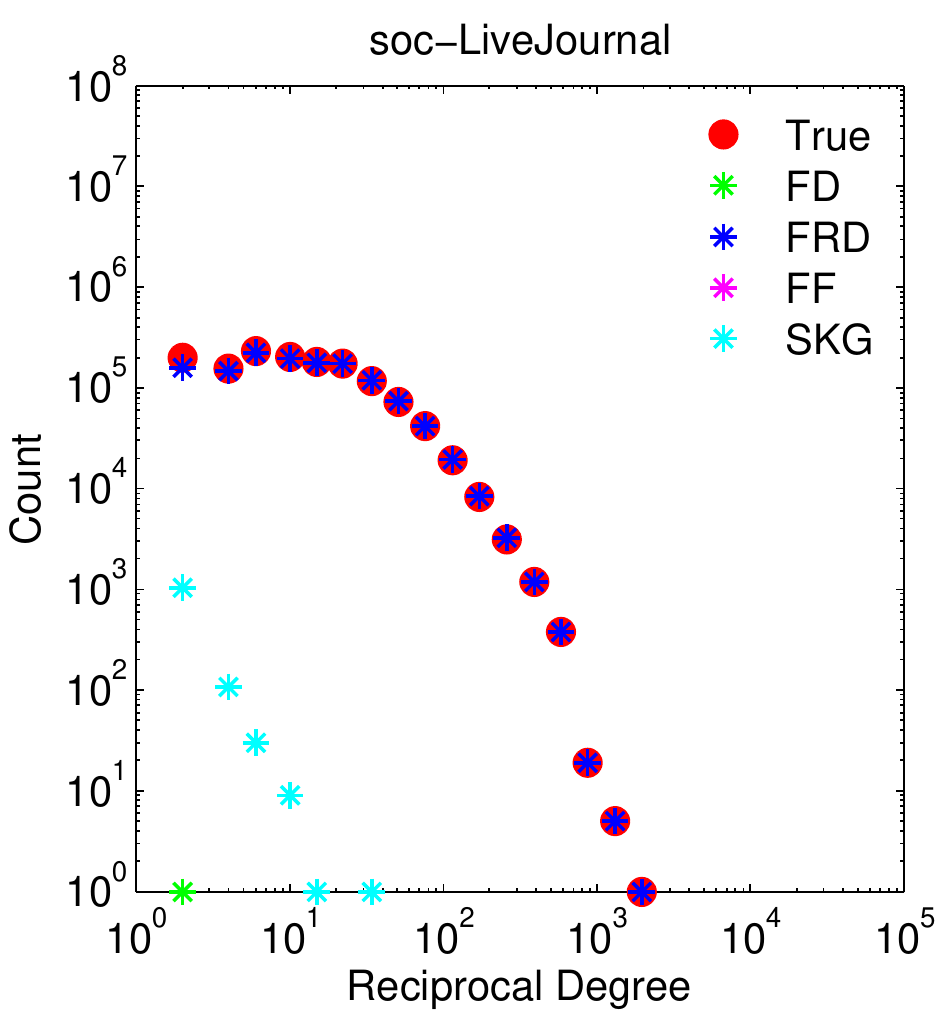}
  }
  \caption{Comparisons of degree distributions produced by various models for graph soc-LiveJournal.}
  \label{fig:soc-LiveJournal}
\end{figure*}
\begin{figure*}[htb]
  \centering
  \subfloat{\label{fig:irdd-youtube}
  \includegraphics[width=1.9in,trim=0 0 0 0]{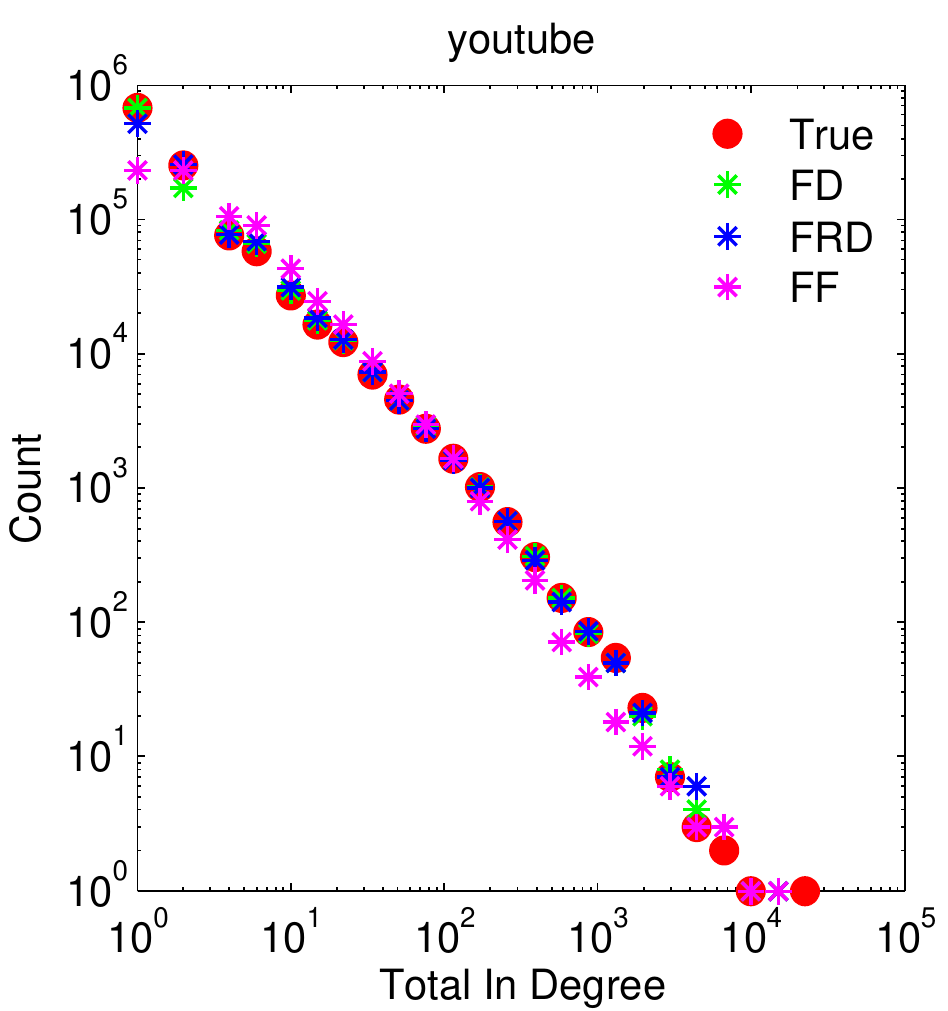}
  }
  \subfloat{\label{fig:ordd-youtube}
    \includegraphics[width=1.9in,trim=0 0 0 0]{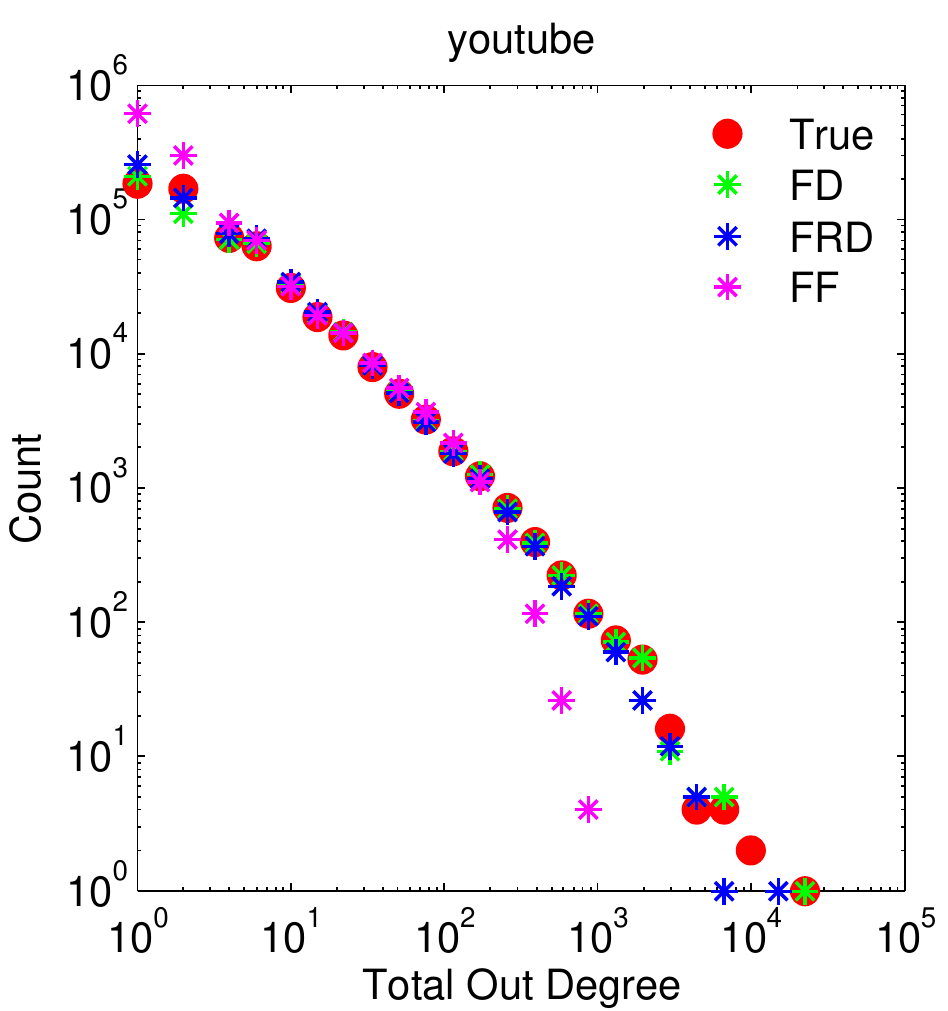}
  }
  \subfloat{\label{fig:rdd-youtube}
    \includegraphics[width=1.9in,trim=0 0 0 0]{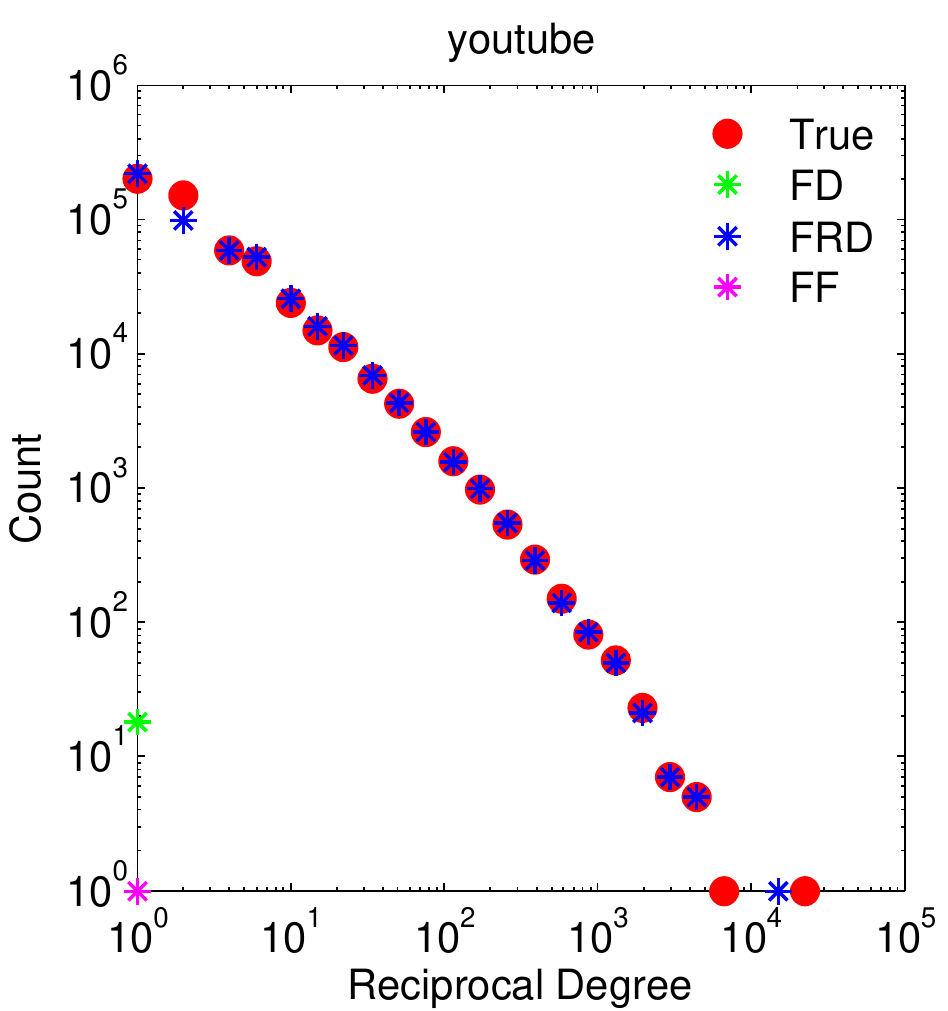}
  }
  \caption{Comparisons of degree distributions produced by various models for graph youtube.}
  \label{fig:youtube}
\end{figure*}
\begin{figure*}[htb]
  \centering
  \subfloat{\label{fig:irdd-flickr}
  \includegraphics[width=1.9in,trim=0 0 0 0]{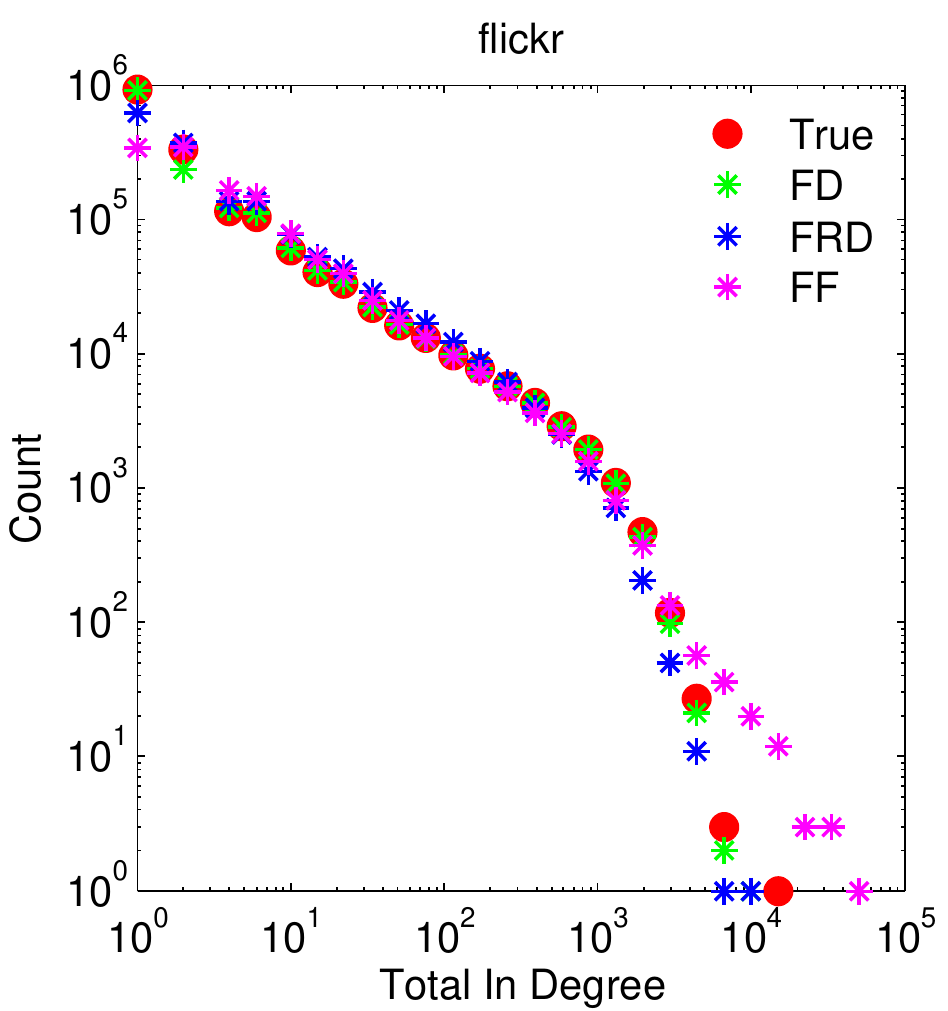}
  }
  \subfloat{\label{fig:ordd-flickr}
    \includegraphics[width=1.9in,trim=0 0 0 0]{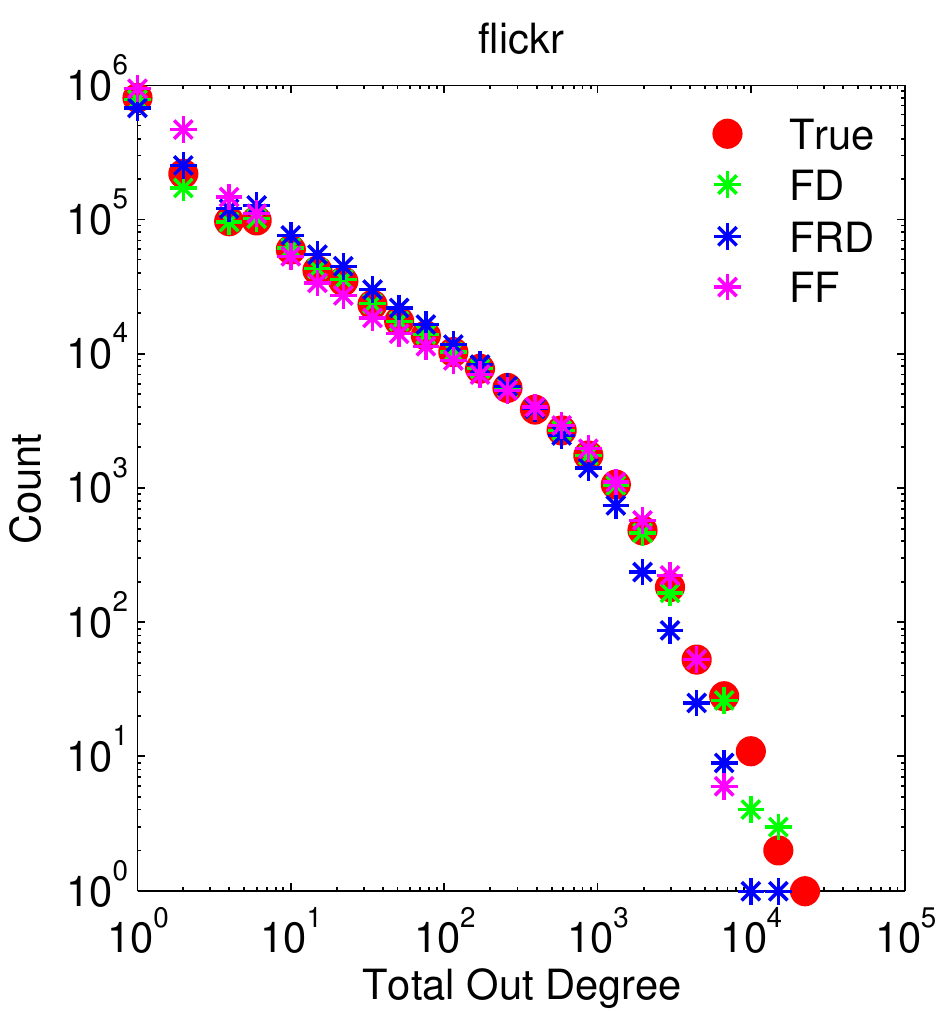}
  }
  \subfloat{\label{fig:rdd-flickr}
    \includegraphics[width=1.9in,trim=0 0 0 0]{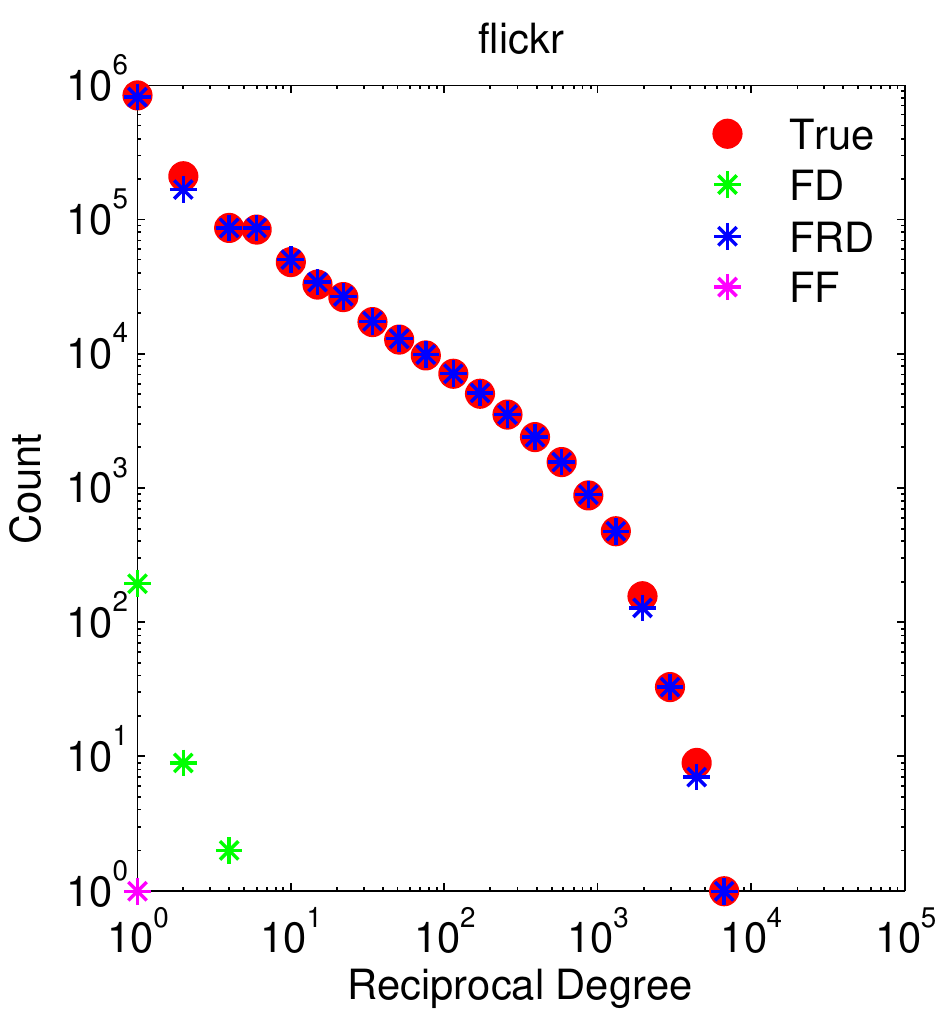}
  }
  \caption{Comparisons of degree distributions produced by various models for graph flickr.}
  \label{fig:flickr}
\end{figure*}
\begin{figure*}[htb]
  \centering
  \subfloat{\label{fig:irdd-livejournal}
  \includegraphics[width=1.9in,trim=0 0 0 0]{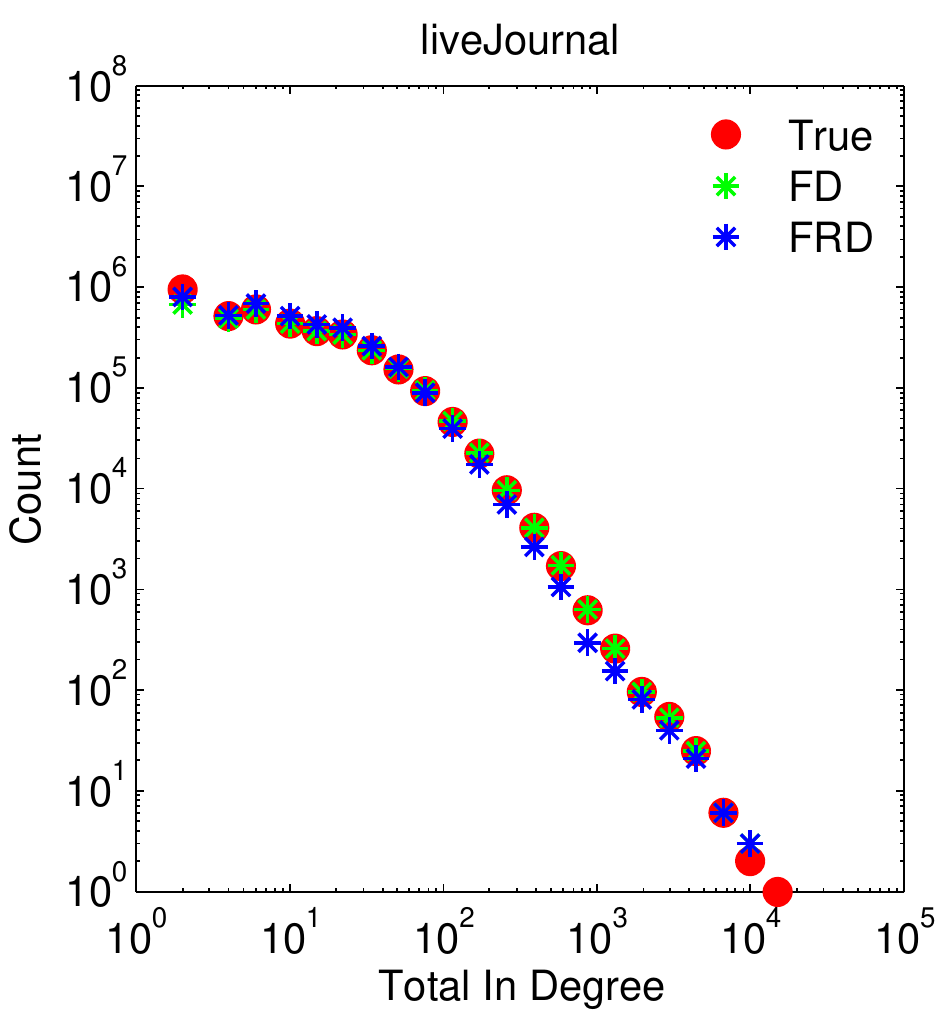}
  }
  \subfloat{\label{fig:ordd-livejournal}
    \includegraphics[width=1.9in,trim=0 0 0 0]{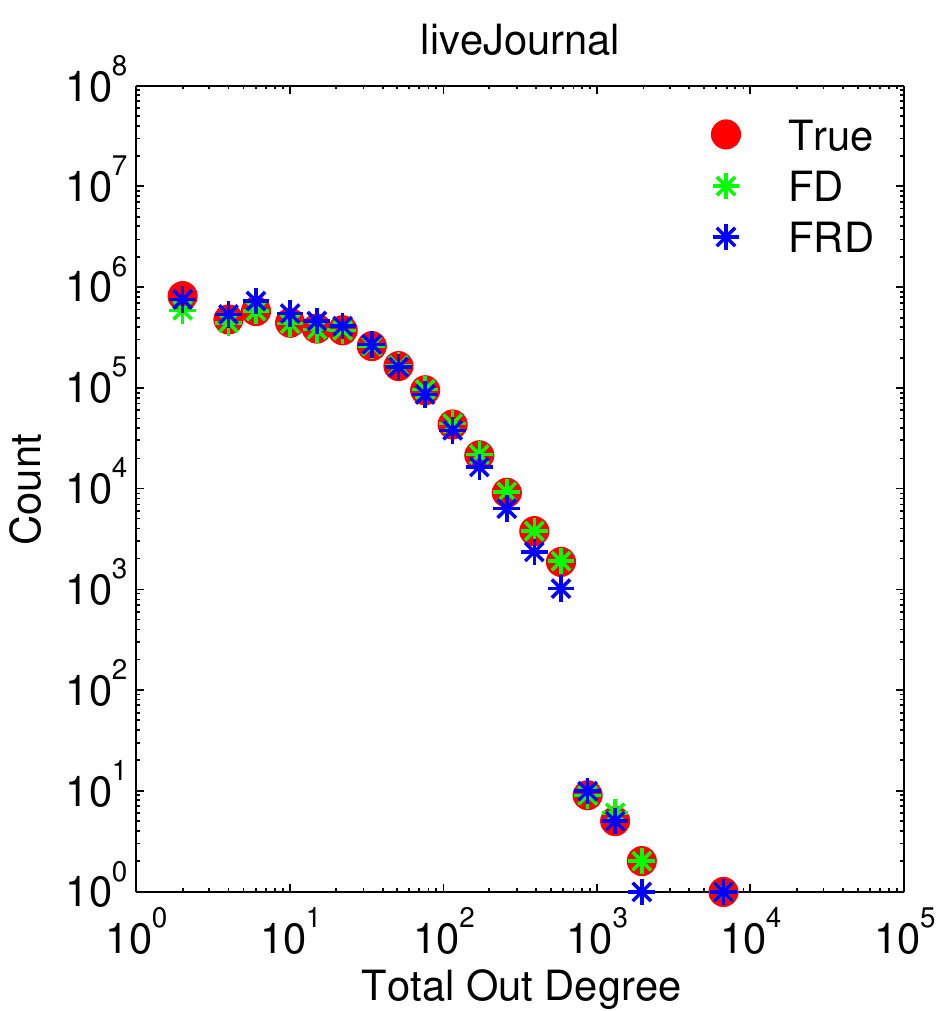}
  }
  \subfloat{\label{fig:rdd-livejournal}
    \includegraphics[width=1.9in,trim=0 0 0 0]{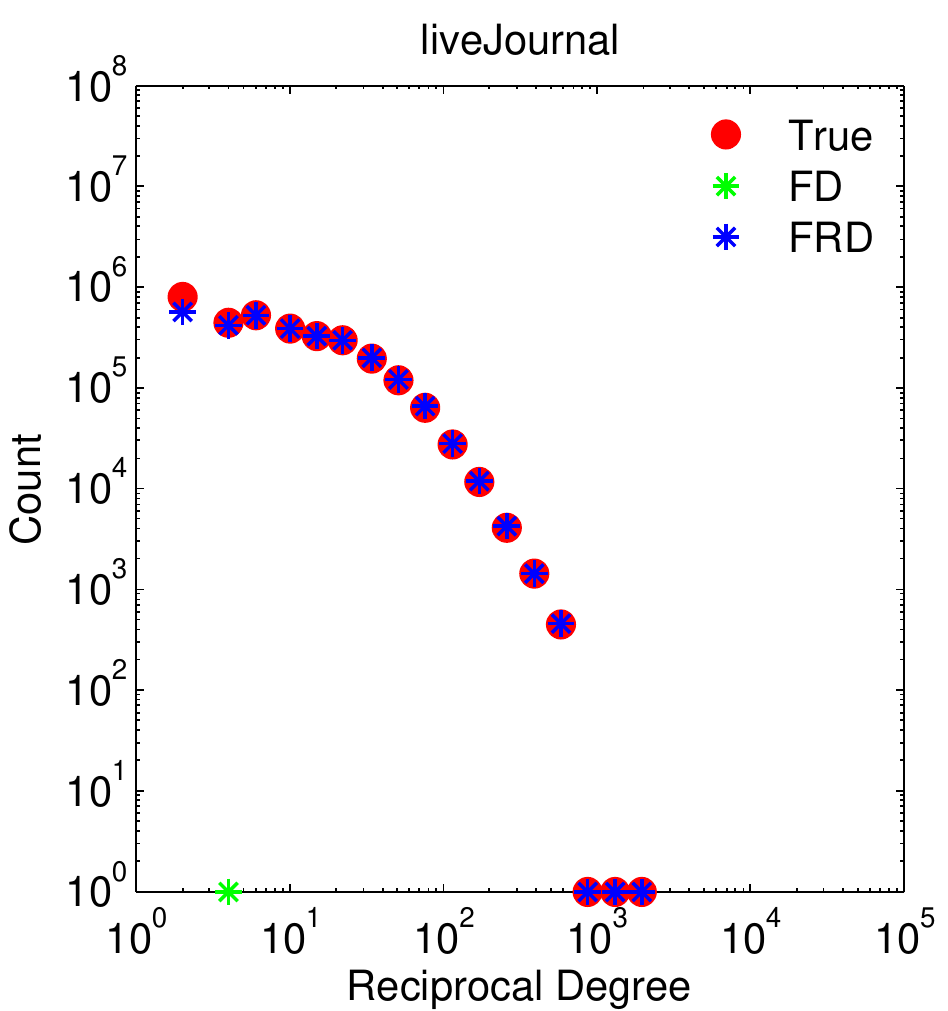}
  }
  \caption{Comparisons of degree distributions produced by various models for graph livejournal.}
  \label{fig:livejournal}
\end{figure*}
\begin{figure*}[htb]
  \centering
  \subfloat{\label{fig:irdd-cit-HepPh}
  \includegraphics[width=1.9in,trim=0 0 0 0]{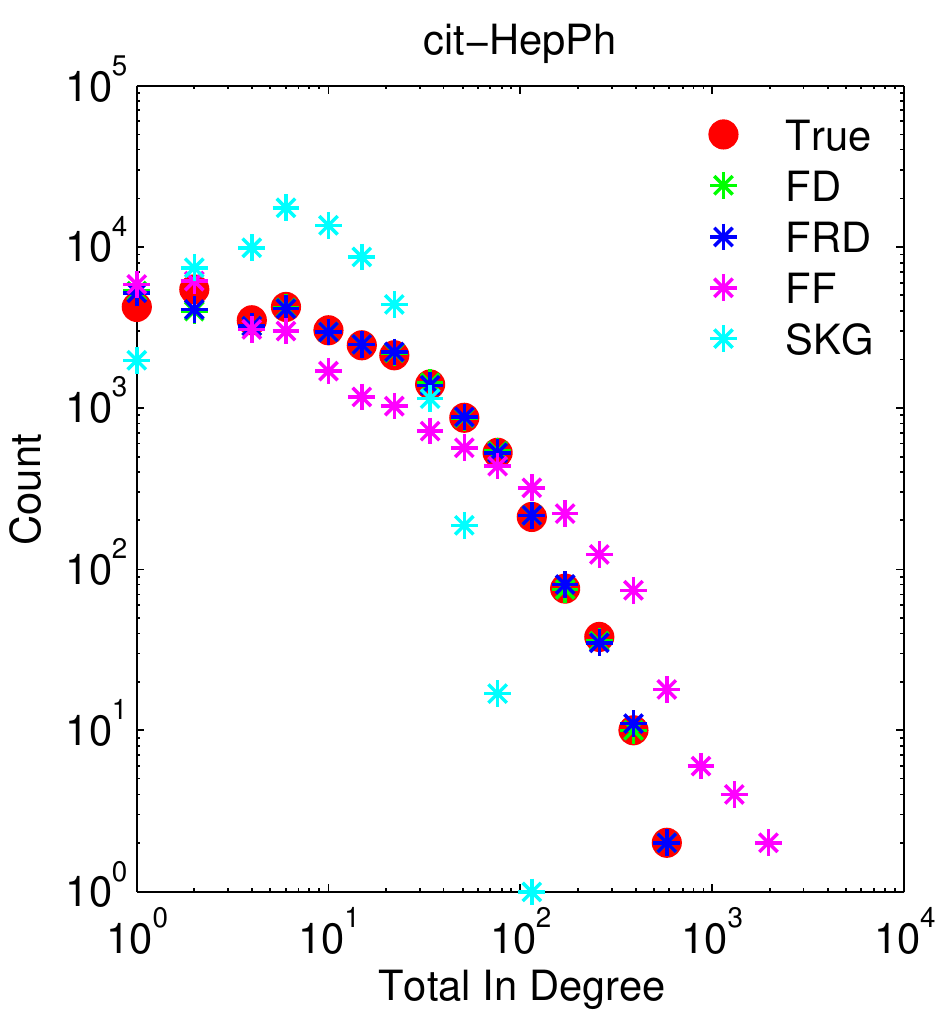}
  }
  \subfloat{\label{fig:ordd-cit-HepPh}
    \includegraphics[width=1.9in,trim=0 0 0 0]{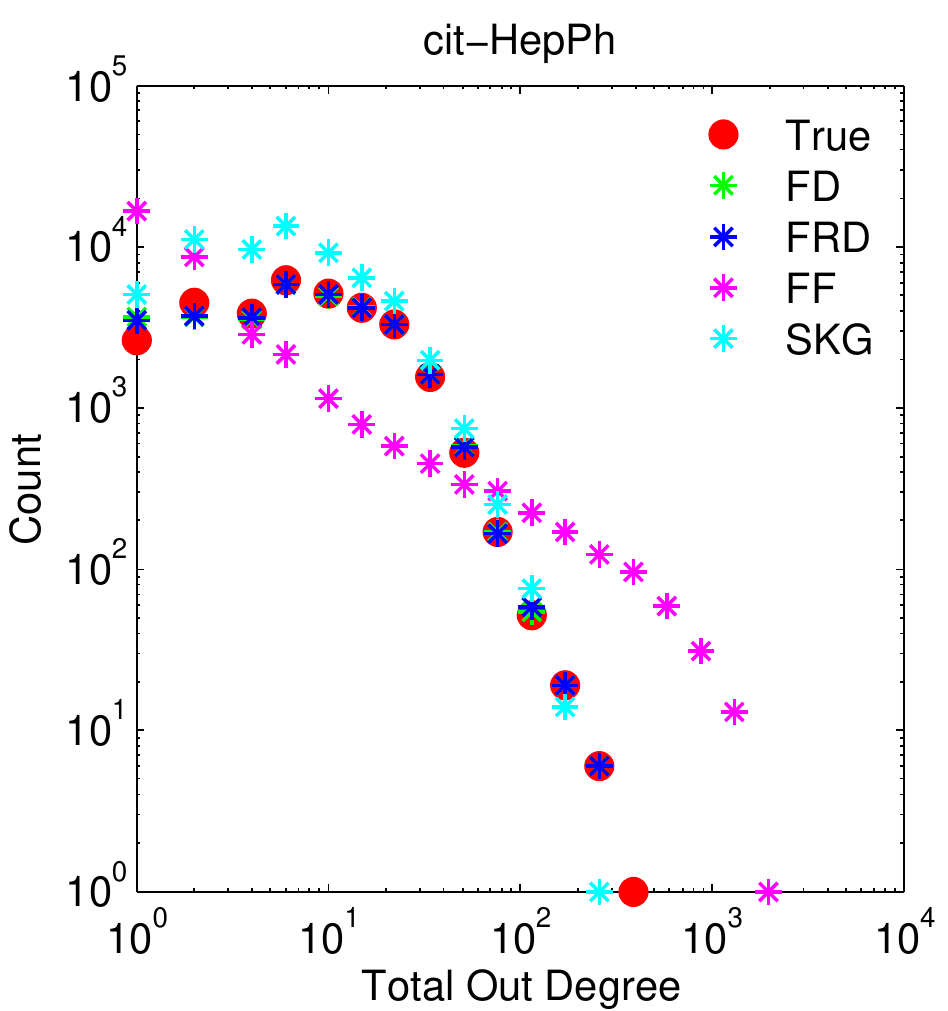}
  }
  \subfloat{\label{fig:rdd-cit-HepPh}
    \includegraphics[width=1.9in,trim=0 0 0 0]{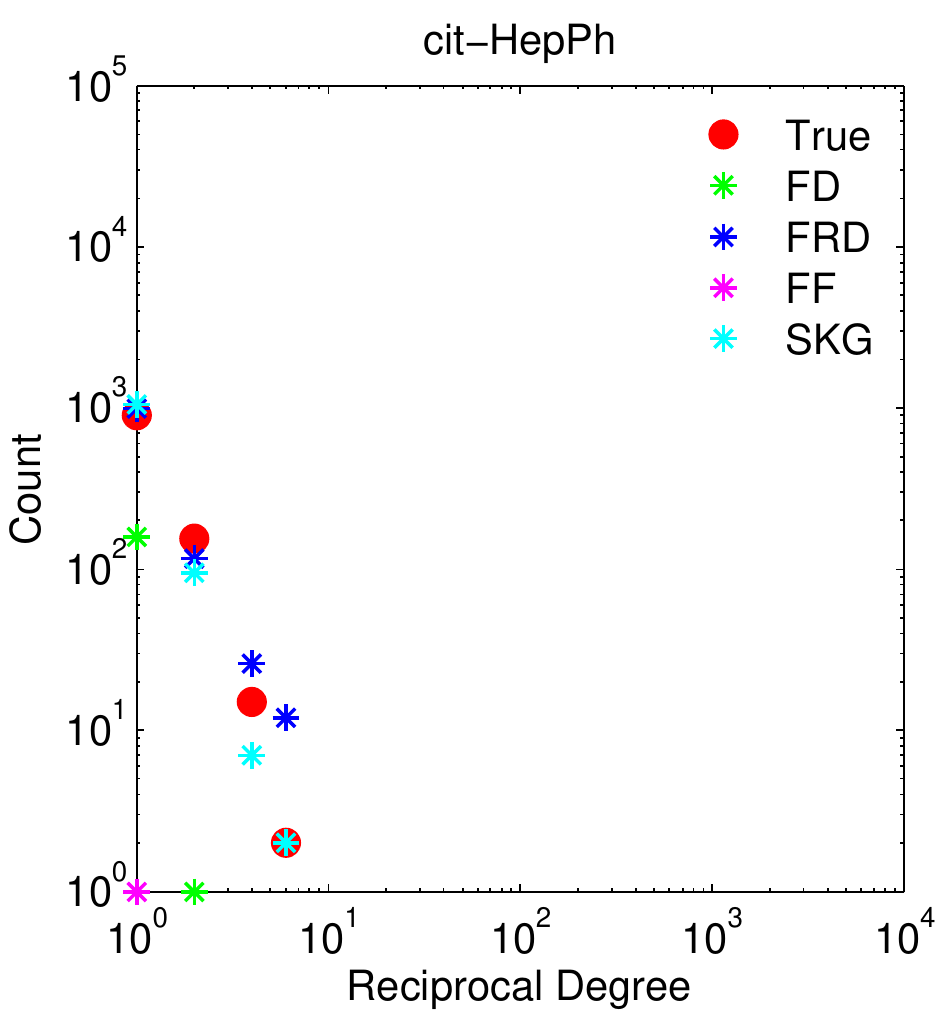}
  }
  \caption{Comparisons of degree distributions produced by various models for graph cit-HepPh.}
  \label{fig:cit-HepPh}
\end{figure*}
\begin{figure*}[htb]
  \centering
  \subfloat{\label{fig:irdd-web-NotreDame}
  \includegraphics[width=1.9in,trim=0 0 0 0]{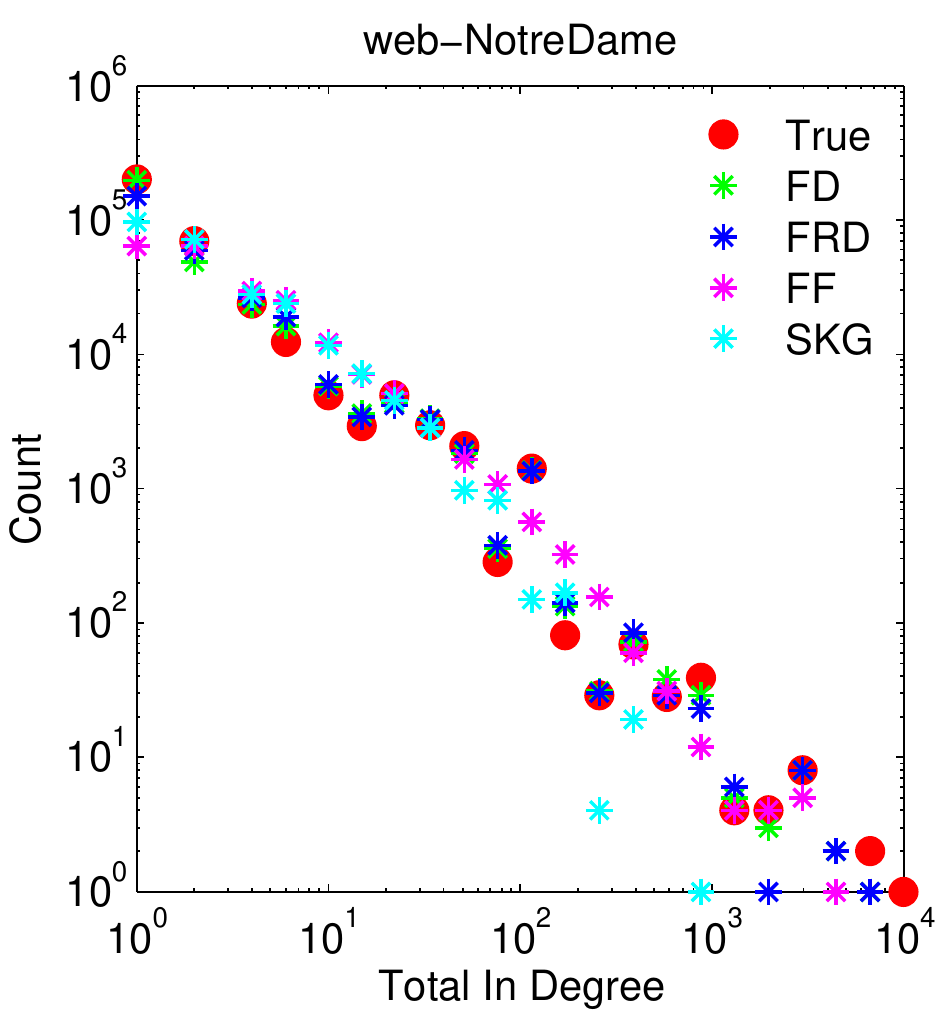}
  }
  \subfloat{\label{fig:ordd-web-NotreDame}
    \includegraphics[width=1.9in,trim=0 0 0 0]{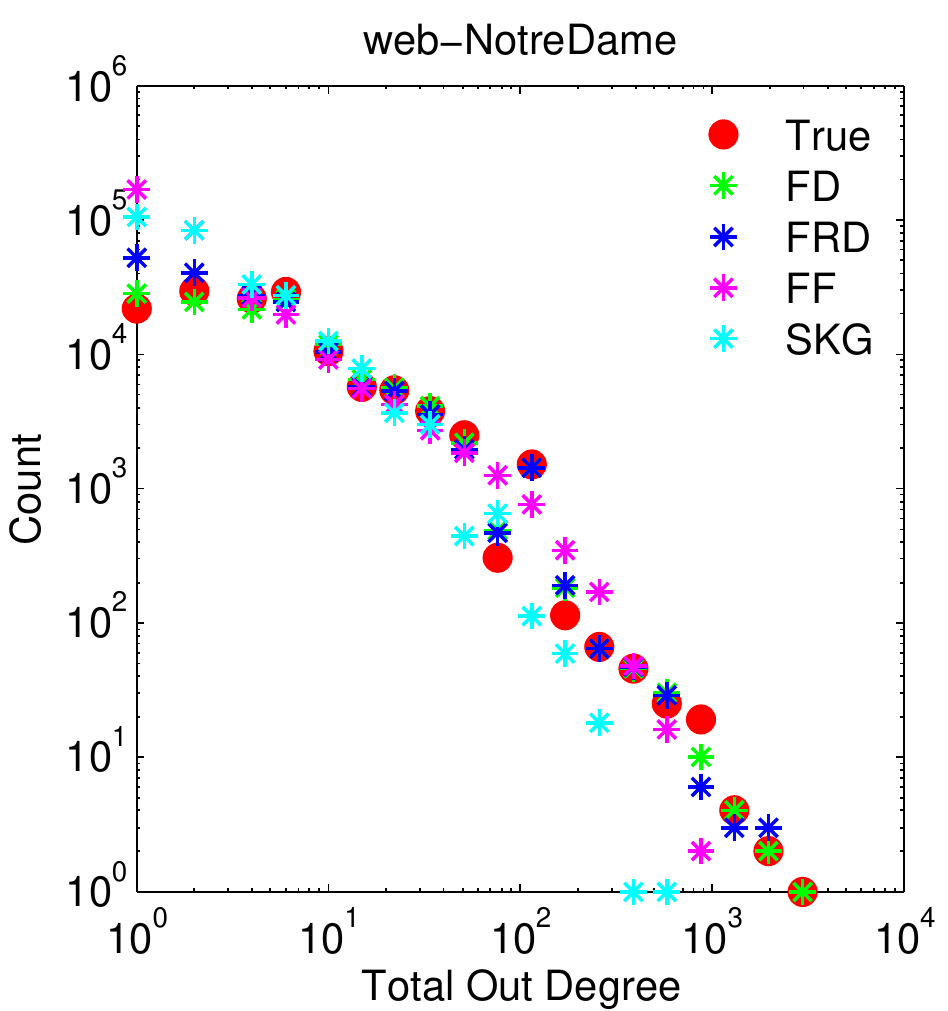}
  }
  \subfloat{\label{fig:rdd-web-NotreDame}
    \includegraphics[width=1.9in,trim=0 0 0 0]{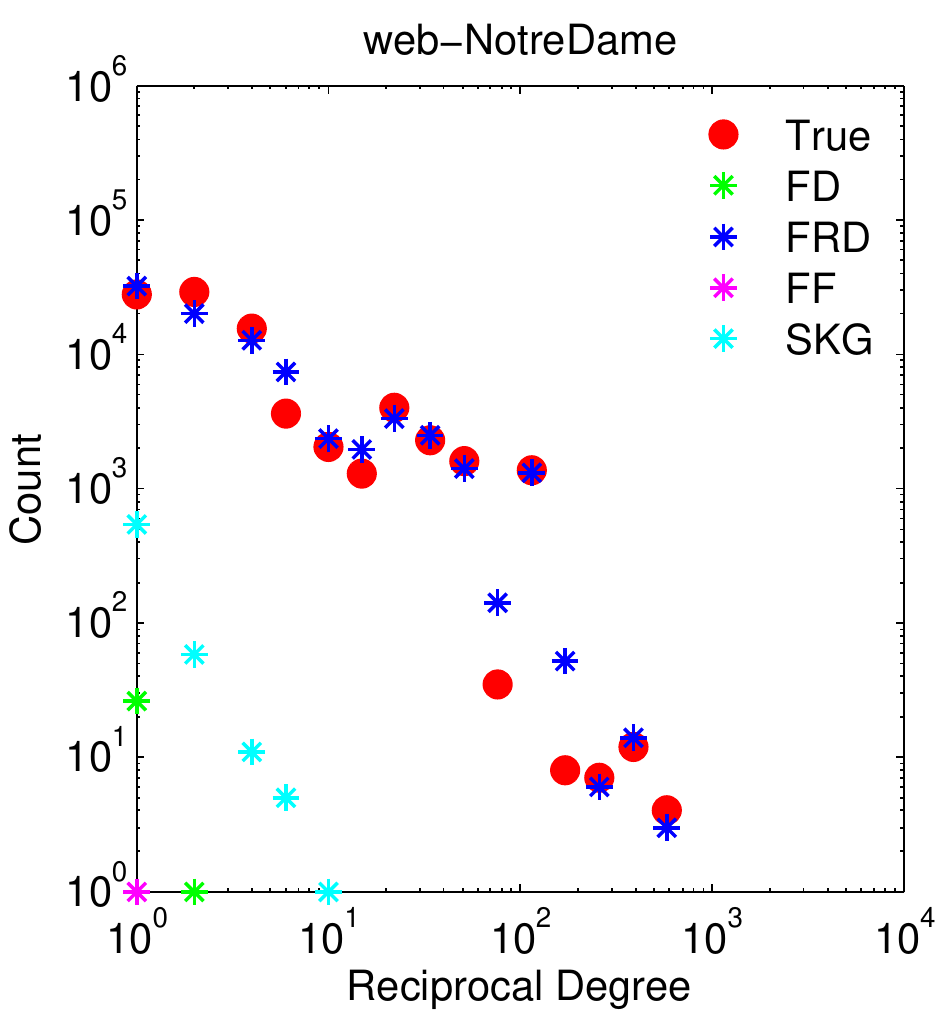}
  }
  \caption{Comparisons of degree distributions produced by various models for graph web-NotreDame.}
  \label{fig:web-NotreDame}
\end{figure*}

We also analyze the generated degree distributions by each model. The
plots are log-binned for  readability.
\Fig{soc-Epinions1} shows the results on the soc-Epinions1
graph. Here we see that all four methods do fairly well in terms of
matching the total in- and out-degree distributions. (The few low
values for SKG are due to its well-known cycling behavior
\cite{CPK11}.) However, only the FRD method matches the reciprocal
degree distribution. The FD and SKG methods produce far too few reciprocal edges and
FF does not produce any.
We see very similar behavior in \Fig{soc-LiveJournal} for
soc-LiveJournal, except here the FF and SKG degree distributions do
not match the total out-degree distribution very well. Once again,
neither FD nor SKG produces many reciprocal edges and FF does not
produce any.

For larger graphs, we have not included SKG due to the expense of
fitting the model. We do compare to FF, however, for the youtube and
flickr graphs shown in \Fig{youtube} and \Fig{flickr}, respectively.
After extensive tuning, FF is able to match the total in- and
out-degree distributions fairly well. But it of course cannot match
the reciprocal degree.
We also show results just for our methods on the largest graph:
livejournal in \Fig{livejournal}. We observe a very close match for
the FRD method in all three distributions. For completeness,
we show results for the citation network cit-HepPh in~\Fig{cit-HepPh} and web network Web-NotreDame in~\Fig{web-NotreDame}.

\section{Conclusion} \label{sec:conc}

Reciprocity in directed networks has not received much attention in terms of
generative models. A first-level goal for a generative model
would be to match  specified  in-, out-, and reciprocal degree distributions. The FRD generator does exactly that and therefore is a good null
model for social network analysis. It is  a variant of Chung-Lu
that explicitly takes care of reciprocal edges.
We find it very intriguing that existing graph models completely ignore reciprocal edges
despite the relatively high fraction of such edges.
While the main challenge in graph modeling would be to design a realistic
model that accounts for reciprocity, we feel that FRD is a first step 
in that direction.

\vspace{-1.5ex}
\section*{Acknowledgment}
This work was funded by the DARPA Graph-theoretic Research
    in Algorithms and the Phenomenology of Social Networks (GRAPHS)
    program and by the DOE Complex  Distributed Interconnected Systems
    (CDIS) Program. Sandia National Laboratories is a multi-program
    laboratory managed and operated by Sandia Corporation, a wholly
    owned subsidiary of Lockheed Martin Corporation, for the
    U.S. Department of Energy's National Nuclear Security
    Administration under contract DE-AC04-94AL85000.

\end{document}